\documentclass[a4paper]{article}

\usepackage{amssymb}
\usepackage{amsthm}
\usepackage{graphicx}
\usepackage{subfigure}
\usepackage{xspace}

\newtheorem{theorem}{Theorem}
\newtheorem{lemma}{Lemma}
\newtheorem{corollary}{Corollary}

\newtheorem{property}{Property}

\newcommand{\SA}{\includegraphics{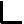}~}
\newcommand{\SB}{\includegraphics{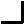}~}
\newcommand{\SC}{\includegraphics{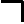}~}
\newcommand{\SD}{\includegraphics{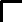}~}
\newcommand{\shapes}{$\{\SA,\SB,\SC,\SD\}$\xspace}

\begin{document}

\title{L-Visibility Drawings of IC-planar Graphs
\thanks{Research supported in part by the MIUR project AMANDA ``Algorithmics for MAssive and Networked DAta''.}}
    
\author{Giuseppe Liotta and Fabrizio Montecchiani\\Universit{\`a} degli Studi di Perugia, Italy\\ \texttt{\small\{giuseppe.liotta,fabrizio.montecchiani\}@unipg.it}} 

\date{}

\maketitle

\begin{abstract}
An IC-plane graph is a topological graph where every edge is crossed at most once and no two crossed edges share a vertex. We show that every IC-plane graph has a visibility drawing where every vertex is of the form \shapes, and every edge is either a horizontal or vertical segment. As a byproduct of our drawing technique, we prove that every IC-plane graph has a RAC drawing in quadratic area with at most two bends per edge.
\end{abstract}


\section{Introduction}

A \emph{visibility drawing} $\Gamma$ of a planar graph $G$ maps the vertices of $G$ into non-overlapping horizontal segments (\emph{bars}), and the edges of $G$ into  vertical segments (\emph{visibilities}), each connecting the two bars corresponding to its two end-vertices. Visibilities intersect bars only at their extreme points. $\Gamma$ is a \emph{strong} visibility drawing if there exists a visibility between two bars if and only if there exists an edge in $G$ between the corresponding vertices. Every biconnected planar graph admits a strong visibility drawing (see, e.g.,~\cite{TamassiaTollis86}). Conversely, if a visibility may not correspond to an edge of the graph, then $\Gamma$ is a {\em weak} visibility drawing. Since every planar graph can be augmented to a biconnected planar graph by adding edges, every planar graph admits a weak visibility drawing. 

The problem of extending visibility drawings to non-planar graphs has been first studied by Dean {\em et al.}~\cite{DBLP:journals/jgaa/DeanEGLST07}. They introduce \emph{bar $k$-visibility drawings}, which are visibility drawings where each bar can see through at most $k$ distinct bars. In other words, each visibility segment can intersect at most $k$ bars, while each bar can be intersected by arbitrary many visibility segments. The graphs that admit a bar $1$-visibility drawing are called \emph{$1$-visibile}. Brandenburg {\em et al.} and independently Evans {\em et al.} prove that \emph{$1$-planar graphs}, i.e., those graphs that can be drawn with at most one crossing per edge, are $1$-visible~\cite{DBLP:journals/jgaa/Brandenburg14,DBLP:journals/jgaa/Evans0LMW14}. They focus on a \emph{weak} model, where there is a visibility through at most $k$ bars if there is an edge, while the converse may not be true. In fact, having a strong model would be too restrictive in the case of bar $k$-visibility drawings. For example, it is easy to see that a cycle of length at least four does not admit a strong bar $1$-visibility drawing~\cite{DBLP:journals/jgaa/Brandenburg14}. In terms of readability, a clear benefit of bar $k$-visibility drawings is that the crossings form right angles. \emph{Right-angle crossing (RAC) drawings} and their advantages in terms of readability have been extensively studied in the graph drawing literature (see, e.g.,~\cite{dl-cargd-12,DBLP:conf/apvis/HuangHE08}). However, in a bar $k$-visibility drawing crossings involve bars and visibilities, i.e., vertices and edges. These crossings are arguably less intuitive than crossings between edges.

Evans {\em et al.} introduce a new model of visibility drawings, called {\em L-visibility drawings}~\cite{elm-svrp+-15}. Their aim is to simultaneously represent two plane $st$-graphs $G_r$ and $G_b$ (whose union might be non-planar). They assume a {\em strong} model, where each vertex is represented by a horizontal bar and a vertical bar that share an extreme point, i.e. it is an {\em L-shape} in the set \shapes.  Two L-shapes are connected by a vertical (horizontal) visibility segment if and only if there exists an edge in $G_r$ ($G_b$) between the corresponding vertices,  no two L-shapes cross one another, and visibilities intersect bars only at their extreme points. A clear advantage of this kind of drawing is that the only possible crossings are between vertical and horizontal visibilites, i.e., between edges of the graph. Furthermore, similar to bar $k$-visibilities, these crossings form right angles. 

\begin{figure}[t]
\centering
\subfigure[]{\includegraphics[scale=0.45,page=1]{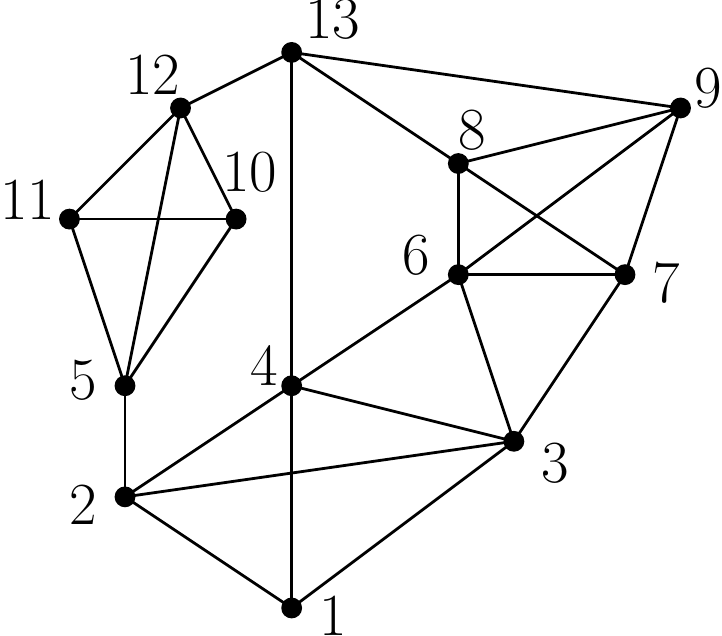}\label{fi:example-1}}\hfil
\subfigure[]{\includegraphics[scale=0.45,page=2]{figures/example}\label{fi:example-2}}\hfil
\subfigure[]{\includegraphics[scale=0.45,page=3]{figures/example}\label{fi:example-3}}
\caption{\small (a) An IC-plane graph $G$. (b) A L-visibility drawing of $G$. (c) A RAC drawing of $G$ with at most two bends per edge.}
\end{figure}

In this paper we initiate the study of {\em weak} L-visibility drawings of non-planar graphs.  We focus on the class of graphs that admit a drawing where each edge is crossed at most once, and no two crossed edges share an end-vertex. These graphs are called \emph{IC-planar graphs} (see  Fig.~\ref{fi:example-1} for an example). Their chromatic number is at most five~\cite{ks-cpgic-JGT10}, and they have at most~$13n/4-6$ edges, which is a tight bound~\cite{zl-spgic-CEJM13}. Recognizing IC-planar graphs is NP-hard~\cite{bdek+-rdicg-15}. Our main contribution is summarized by the following theorem, proved in Section~\ref{se:proof}. See Fig.~\ref{fi:example-2} for an example of a drawing computed by using Theorem~\ref{th:main}.

\begin{theorem}\label{th:main}
Every $n$-vertex IC-plane graph $G$ admits a L-visibility drawing in $O(n^2)$ area, which can be computed in $O(n)$ time.
\end{theorem}

We remark that Theorem~\ref{th:main} contributes to the rapidly growing literature devoted to the problem of drawing graphs that are ``nearly planar'' in some sense, i.e. graphs where only some types of edge crossings are allowed (for example, an edge can be crossed at most a constant number of times); see e.g.,~\cite{DBLP:conf/ictcs/Liotta14} for references.  
In particular,  Brandenburg {\em et al.} have recently described a cubic-time algorithm that computes IC-planar drawings with right-angle crossings and straight-line edges~\cite{bdek+-rdicg-15}. However these drawings may require exponential area, which is proved to be worst-case optimal~\cite{bdek+-rdicg-15}.  Brandenburg {\em et al.} leave as an open problem to study techniques that compute IC-planar drawings in polynomial area and with good crossing resolution~\cite{bdek+-rdicg-15}. We also recall that every graph admits a RAC drawing with at most three bends per edge~\cite{DBLP:journals/tcs/DidimoEL11}, while testing whether a graph has a straight-line RAC drawing is NP-hard~\cite{DBLP:journals/jgaa/ArgyriouBS12}. The following corollary follows as a byproduct of Theorem~\ref{th:main} (see also Fig.~\ref{fi:example-3}).

\begin{corollary}\label{co:main}
Every $n$-vertex IC-plane graph $G$ admits a RAC drawing with at most two bends per edge in $O(n^2)$ area, which can be computed in $O(n)$ time.
\end{corollary}

\section{Preliminaries}\label{se:preliminaries}
We assume familiarity with basic graph drawing concepts, see also~\cite{dett-gd-99}. 

{\bf Planarity and connectivity.} A graph $G=(V,E)$ is \emph{simple}, if it contains neither loops nor multiple edges. We consider simple graphs, if not otherwise specified. A \emph{drawing} $\Gamma$ of $G$ maps each vertex of $V$ to a point of the plane and each edge of $E$ to a Jordan arc between its two end-points. We only consider \emph{simple drawings}, i.e., drawings such that the arcs representing two edges have at most one point in common, which is either a common end-vertex or a common interior point where the two arcs properly cross. A drawing is \emph{planar} if no two arcs representing two edges cross. A planar drawing divides the plane into topologically connected regions, called \emph{faces}. The unbounded region is called the \emph{outer face}. A \emph{planar embedding} of a graph is an equivalence class of planar drawings that define the same set of faces. A graph with a given planar embedding is a \emph{plane} graph. For a non-planar drawing, we can still talk about embedding considering that the boundary of a face may consist of portions of arcs between vertices and/or crossing points. 

A graph is \emph{biconnected} if it remains connected after removing any one vertex. A directed graph (a digraph for short) is biconnected if its underlying undirected graph is biconnected. A \emph{topological numbering} of a digraph is an assignment, $X$, of numbers to its vertices such that $X(u) < X(v)$ for every edge $(u,v)$. A graph admits a topological numbering if and only if it is acyclic. An acyclic digraph with a single source $s$ and a single sink $t$ is called an \emph{$st$-graph}. A \emph{plane $st$-graph} is an $st$-graph that is planar and embedded such that $s$ and $t$ are on the boundary of the outer face. In any $st$-graph, the presence of the edge $(s,t)$ guarantees that the graph is biconnected. In the following we consider $st$-graphs that contain the edge $(s,t)$, as otherwise it can be added without violating planarity. Let $G$ be a plane $st$-graph, then for each vertex $v$ of $G$ the incoming edges appear consecutively around $v$, and so do the outgoing edges. Vertex $s$ only has outgoing edges, while vertex $t$ only has incoming edges. This particular transversal structure is known as a \emph{bipolar orientation}~\cite{DBLP:journals/dcg/RosenstiehlT86,TamassiaTollis86}. Each face $f$ of $G$ is bounded by two directed paths with a common \emph{origin} and \emph{destination}, called the \emph{left path} and \emph{right path} of $f$. 

\begin{figure}[t]
    \centering
    \subfigure[]{\includegraphics[scale=0.8,page=1]{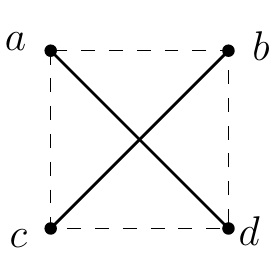}\label{fi:thomassen-1}}
    \hfil
    \subfigure[]{\includegraphics[scale=0.8,page=2]{figures/thomassen}\label{fi:thomassen-2}}
  \caption{\small (a) An X-configuration and (b) a B-configuration.}
\end{figure}

{\bf IC-planar graphs.} We recall some definitions also given in~\cite{bdek+-rdicg-15}. A drawing is IC-planar if each edge is crossed at most once, and any two crossed edges do not share an end-vertex. See Fig.~\ref{fi:example-1} for an illustration. An \emph{IC-planar embedding} is an embedding derived from an IC-planar drawing. A graph with a given IC-planar embedding is an \emph{IC-plane} graph. Thomassen~\cite{t-rdg-JGT88} characterized the possible crossing configurations that occur in a 1-planar drawing, i.e., a drawing where each edge is crossed at most once. This characterization applied to IC-planar drawings gives rise to the following property, where an X-crossing is of the  type described in Fig.~\ref{fi:thomassen-1}, and a B-crossing is of the type described in Fig.~\ref{fi:thomassen-2} (the solid edges only). 

\begin{property}[\cite{bdek+-rdicg-15}]\label{pr:char-crossins}
  Every crossing of an IC-plane graph is either an X- or a B-crossing.
\end{property}
A \emph{kite}~$K$ is a graph isomorphic to~$K_4$ together with an embedding such that all the vertices are on the boundary of the outer face. This implies that two edges of~$K$ cross each other, while the other four edges are not crossed and  belong to the boundary of the outer face; see Fig.~\ref{fi:thomassen-1}. Consider a pair of crossing edges of an IC-plane graph $G$, such that their four end-vertices induce a kite $K$. The kite $K$ is \emph{empty}, if in $G$ there is no other vertex inside the internal faces of $K$. The following property is a consequence of the more general Lemma 1 in~\cite{bdek+-rdicg-15} (in particular of cases c1 and c2 of that lemma). 

\begin{property}[\cite{bdek+-rdicg-15}]\label{pr:augmentation}
Let $G=(V,E)$ be an $n$-vertex IC-plane graph. It is possible to augment $G$ to a biconnected IC-plane graph $G^+=(V,E^+)$ (with a possibly different embedding), where $E \subseteq E^+$, such that the end-vertices of each pair of crossing edges of $G^+$ induce an empty kite. This can be done in $O(n)$ time.
\end{property}

{\bf Visibility model.}  In a L-visibility drawing $\Gamma$ of a graph $G$, every vertex is represented by a horizontal and a vertical segment sharing an end-point, i.e., by an L-shape in the set \shapes. Each edge of $G$ is drawn in $\Gamma$ as either a horizontal or a vertical visibility segment joining the two L-shapes corresponding to its two end-vertices. Clearly, horizontal visibilities only cross vertical visibilities at right angles. Also, no two L-shapes intersect. If $G$ is an IC-plane graph, then each visibility is crossed at most once and no two crossed visibilities are incident to the same L-shape. In Fig.~\ref{fi:example-2}, an L-visibility representation $\Gamma$ of an IC-plane graph $G$ is shown.  Finally, we adopt a weak model, where a visibility may not imply the existence of the corresponding edge in the graph. For example, in Fig.~\ref{fi:example-2} the L-shapes of vertices $8$ and $10$ can be joined by a horizontal visibility, although the edge $(8,10)$ does not exist in $G$.

\section{Proof of Theorem~\ref{th:main}}\label{se:proof}

The proof of Theorem~\ref{th:main} is constructive and is based on a drawing algorithm that takes as input an IC-plane graph $G$ and returns a L-visibility drawing $\Gamma$ of $G$. By Property~\ref{pr:augmentation}, we assume that $G$ is such that each crossing induces an empty kite (see Section~\ref{se:preliminaries}). In fact, the output of our drawing algorithm maintains the IC-planar embedding obtained by applying Property~\ref{pr:augmentation}. We begin by removing from $G$ all pairs of crossing edges and orient the resulting graph $P$ to an $st$-graph. The computed orientation is such that, when reinserting a pair of crossing edges in the corresponding planar face of $P$, one of them is always incident to the origin and the destination of the face. In other words, each face of $P$ that corresponds to an empty kite of $G$, is oriented so that its left and right paths contain exactly one vertex each. To prove that this is always the case, we first need to introduce additional notation. Let $f$ be a face of a plane graph $G$.  Let $v_1,\dots,v_h$ be the $h\geq3$ vertices that belong to the boundary of $f$, and let $N(v)$ be the set of neighbors of a vertex $v$ of $G$. The \emph{contraction} of $f$ is the  operation defined as follows. Add to $G$ a vertex $v_f$ and connect $v_f$ to the vertices in $N(v_1) \cup \dots \cup N(v_h)$. Then remove $v_1,\dots,v_h$ from $G$. The resulting (multi)graph is still planar. Moreover, the contraction operation can be performed so to preserve the planar embedding of $G$. Namely, the circular order of the edges incident to $v_f$ is the same circular order encountered walking along the boundary of $f$. See also Fig.~\ref{fi:contraction} for an illustration. The original graph $G$ can be obtained by applying the reverse operation, called the \emph{expansion} of $v_f$. Namely, vertices $v_1,\dots,v_h$ are reinserted along with their original edges and $v_f$ is removed from the graph.

\begin{figure}[t]
\centering
\subfigure[ Contraction of $f_K$]{\includegraphics[scale=0.5,page=1]{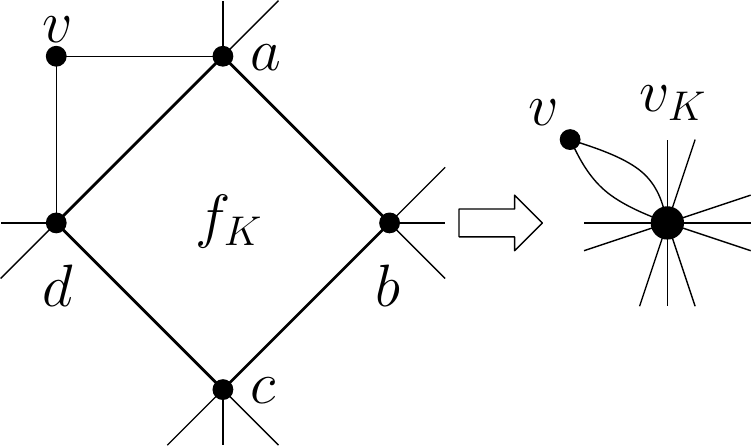}\label{fi:contraction}}\hfil
\subfigure[{\bf Case 1c}]{\includegraphics[scale=0.5,page=2]{figures/contraction}\label{fi:expansion-1}}\hfil
\subfigure[{\bf Case 2b}]{\includegraphics[scale=0.5,page=3]{figures/contraction}\label{fi:expansion-2}}\hfil
\subfigure[{\bf Case 3a}]{\includegraphics[scale=0.5,page=4]{figures/contraction}\label{fi:expansion-3}}
\caption{\small Illustration for the proof of Lemma~\ref{le:st-orientation}.}
\end{figure}

\begin{lemma}\label{le:st-orientation}
Let $G$ be an $n$-vertex IC-plane graph such that the end-vertices of each pair of crossing edges induce an empty kite. Let $C$ be the set of crossing edges in $G$, and consider the plane graph $P = G\setminus C$. Graph $P$ can be oriented to an $st$-graph such that each pair of crossing edges of $C$ has been removed from a face of $P$ whose left and right paths contain exactly one vertex each. This operation can be done in $O(n)$ time.
\end{lemma}
\begin{proof}
Each pair of crossing edges of $G$ induces an empty kite $K$, and thus corresponds to a single face $f_K$ in $P$ having exactly four vertices on its boundary. As a first step, we contract each face $f_K$ of $P$ (corresponding to a kite $K$ in $G$). See also Fig.~\ref{fi:contraction} for an illustration. Notice that, since $G$ is an IC-plane graph, no two faces of $P$ share a vertex, and thus each vertex $v_K$ of $P_C$ corresponds to exactly one face of $P$. Hence, we can contract the faces following an arbitrary order. The resulting graph $P_C$ is a plane (multi)graph. Indeed, observe that if a face $f_K$ of $P$ shares an edge with a triangular face $f$, then $P_C$ will contain two parallel edges between $v_K$ and the vertex $v$ of $f$ not in $f_K$ (see also Fig.~\ref{fi:contraction}). 

As a second step, we orient $P_C$ to an $st$-graph (observe that parallel edges must receive the same orientation). In the third step, we expand one by one all the vertices $v_K$ corresponding to a contracted face $f_K$. After expanding a vertex $v_K$, we orient the four reinserted edges of the face $f_K$  maintaining the following invariants: {\bf I1.} The resulting graph has a single source and a single sink; {\bf I2.} The left and right paths of $f_K$ contain exactly one vertex each.

Invariants {\bf I1} and {\bf I2} imply that the graph after expanding $v_K$ is still an $st$-graph, as it has a single source and a single sink by {\bf I1} and is acyclic by {\bf I2}. 

Let $\{a,b,c,d\}$ be the four vertices belonging to the boundary of face $f_K$, encountered in this order clockwise around the boundary of the face. To maintain {\bf I2} we need to orient the edges of $f_K$ such that the origin and the destination of $f_K$ are two non-adjacent vertices, i.e., either $a$ and $c$ or $b$ and $d$. In order to maintain {\bf I1}, recall that the incoming edges of $v_K$ appear consecutive around it and so do the outgoing edges, unless $v_K$ is the source or the sink of the graph. Thus, if $v_K$ is neither the source nor the sink of the graph, then at most two of the reinserted vertices will have both incoming and outgoing edges incident to $v_K$, whereas at most three will have only incoming or only outgoing  edges incident to $v_K$. We distinguish the following three cases.

\medskip

{\noindent \bf Case 1.} No vertex of $f_K$ has both incoming and outgoing edges. See Fig.~\ref{fi:expansion-1} for an illustration. Then we consider the following subcases. {\bf Case 1a.} There are three vertices having only outgoing edges, say $a$, $b$ and $c$, which implies that all the edges of $d$ are incoming. In this case, we orient the edges of $f_K$ so that $b$ is the destination and $d$ is the origin of the face, which ensures {\bf I2}. All vertices has now both incoming and outgoing edges, and thus {\bf I1} is also maintained. 
{\bf Case 1b.} There are three vertices having only incoming edges, say $a$, $b$ and $c$, which implies that all the edges of $d$ are outgoing. In this case, the orientation that ensures {\bf I1} and {\bf I2} is the one where $d$ is the destination and $b$ is the origin of the face. 
{\bf Case 1c.} Two (consecutive) vertices only have incoming edges, say $c$ and $d$, and two (consecutive) vertices only have outgoing edges, $a$ and $b$.  Then to maintain {\bf I1} and {\bf I2} we let $a$ to be the destination and $c$ the origin of the face. This particular case is shown in Fig.~\ref{fi:expansion-1}. 

\medskip

{\noindent \bf Case 2.} Only one vertex of $f_K$ , say $a$, has both incoming and outgoing edges. Moreover, assume that the incoming edges of $a$ are between the edge $(a,d)$ and the outgoing edges of $a$, as in Fig.~\ref{fi:expansion-2}, since the case in which the incoming edges of $a$ are between the outgoing edges of $a$ and the edge $(a,b)$ is symmetric. We have two subcases. 
{\bf Case 2a.} The other three vertices only have incoming (resp., outgoing) edges. In this case,  we orient the edges of $f_K$ so that $c$ (resp., $a$) is the origin and $a$ (resp., $c$) is the destination. This choice ensures both {\bf I1} and {\bf I2}. 
{\bf Case 2b.} Two vertices only have incoming edges and the other vertex only have outgoing edges. Due to the bipolar orientation of $v_K$, the two vertices having only incoming edges are $c$ and $d$. Then we choose $d$ as origin of the face and $b$ as destination, which ensures both {\bf I1} and {\bf I2}. This particular case is shown in Fig.~\ref{fi:expansion-2}. The same orientation works also if two vertices only have outgoing edges ($b$ and $c$) and the other vertex only have incoming edges ($d$).

\medskip

{\noindent \bf Case 3.} Two vertices of $f_K$  have both incoming and outgoing edges. We have two subcases. {\bf Case 3a.} Suppose first that these two vertices, say $a$ and $b$, are adjacent in $P$.   Moreover, assume that the incoming edges of $a$ are between the edge $(a,d)$ and the outgoing edges of $a$, as in Fig.~\ref{fi:expansion-3}, since the other case is symmetric. This implies that the incoming edges of $b$ are between the outgoing edges of $b$ and the edge $(b,c)$. Moreover, $c$ and $d$ only have incoming edges. Then, if we let $a$ to be the destination of the face and $c$ the origin, {\bf I1} and {\bf I2} are maintained. 
{\bf Case 3b.} Suppose now that the two vertices, say $a$ and $c$, are not adjacent in $P$. Moreover, assume that the incoming edges of $a$ are between the edge $(a,d)$ and the outgoing edges of $a$, as in Fig.~\ref{fi:expansion-3}, since the case in which the incoming edges of $a$ are between the outgoing edges of $a$ and the edge $(a,b)$ is symmetric. This implies that the incoming edges of $c$ are between the outgoing edges of $c$ and the edge $(c,d)$. Moreover, $b$ and $d$ only have outgoing and incoming edges, respectively. Then, if we orient the edges of $f_K$ so that $b$ is the destination of the face and $d$ the origin, then {\bf I1} and {\bf I2} are ensured.

\medskip

Finally, suppose that $v_K$ is the source (resp., sink) of the graph. Then all the vertices of $f_K$ have either no incident oriented edges, or all outgoing (resp., incoming) edges. Also, at least one of them has at least one outgoing (resp., incoming) edge, say $a$. We orient the edges of $f_K$ so that $a$ is the destination (resp., origin) and $c$ the origin (resp., destination) of $f_K$. 

\medskip

The described algorithm works in $O(n)$ time. Namely, the graph $P_C$ can be constructed and oriented to an $st$-graph in $O(n)$ time (see, e.g.,~\cite{Even1976339}). Furthermore, orienting the edges of an expanded face requires first to analyze the orientation of the edges incident to each vertex of the face, and then to orient the edges of the face. Since the faces are vertex-disjoint, this costs at most $\sum_{\forall v \in P}deg(v)=2m_P$, where $m_P$ is the number of edges of $P$, which is $O(n)$.
\end{proof}

Let $f_K$ be a face of the plane $st$-graph $P$ which corresponds to an (empty) kite $K$ in $G$. By Lemma~\ref{le:st-orientation}, $f_K$ is such that its left and right path have both length two. This implies that one of the two crossing edges of $K$ is incident to the origin and to the destination of $f_K$, whereas the other one is incident to the two vertices belonging one to the left and one to the right path of $f_K$. Let $P^+$ be the biconnected plane graph obtained from $P$ by reinserting, for each pair of crossing edges of $G$, the edge incident to the origin and to the destination of the corresponding face in $P$. Furthermore, let $\Gamma'$ be a strong visibility drawing of $P^+$ in $O(n^2)$ area, which can be computed in $O(n)$ time (see, e.g.,~\cite{TamassiaTollis86}).

\begin{figure}[t]
\centering
\subfigure[$f_k$]{\includegraphics[scale=0.8,page=3]{figures/visrep}\label{fi:face}}\hfil
\subfigure[$\Gamma'$]{\includegraphics[scale=0.8,page=1]{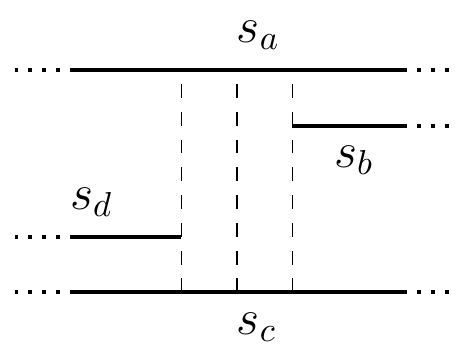}\label{fi:barvisrep}}\hfil
\subfigure[$\Gamma$]{\includegraphics[scale=0.8,page=2]{figures/visrep}\label{fi:lvisrep}}\hfil
\caption{\small Illustration for the proof of Lemma~\ref{le:drawing}.}
\end{figure}

\begin{lemma}\label{le:drawing}
$\Gamma'$ can be transformed into a L-visibility drawing $\Gamma$ of $G$ that requires $O(n^2)$ area. This operation can be done in $O(n)$ time.
\end{lemma}
\begin{proof}
Consider a face $f_K$ of $P$ corresponding to an empty kite $K$ in $G$. Let $\{a,b,c,d\}$ be the four vertices of $f_K$, encountered in this order clockwise around the boundary of the face. See also Fig.~\ref{fi:face}. Without loss of generality, let $a$ be the destination of the face. Then Lemma~\ref{le:st-orientation} implies that $c$ is the origin of $f_K$, and thus the edge reinserted in $P^+$ for this face is $(a,c)$. In other words, $f_K$ is split in two faces in $P^+$, and these two faces share the edge $(a,c)$. Consider the subdrawing of $\Gamma'$ induced by the four vertices $\{a,b,c,d\}$ . An illustration is also shown in Fig.~\ref{fi:barvisrep}. Let $s_v$ be the bar representing a vertex $v$ in $\Gamma'$. Either $s_d$ and $s_b$ are drawn at the same $y$-coordinate, or, one is above the other. Also, the two bars do not overlap as there is no edge between $b$ and $d$ in $P^+$ (and $\Gamma'$ is a strong visibility drawing). Between the two bars there is actually (at least) one unit gap needed to draw the visibility from $s_c$ to $s_a$. Moreover, $s_c$ is below both $s_d$ and $s_b$, while $s_a$ is above both of them. In any case, we first extend $s_d$ (resp., $s_b$) by 0.25 units to the right (resp., left). Next, if $s_b$ and $s_d$ have the same $y$-coordinate, then it suffice to draw two vertical bars $s'_d$ and $s'_b$, such that the bottomost end-point of $s'_d$ (resp., $s'_b$) coincides with the rightmost end-point of $s_d$ (resp., leftmost end-point of $s_b$), and such that the other end-point is 0.5 units above it.  If one is above the other, say $s_b$ is above $s_d$, and the difference in terms of $y$-coordinates between the two bars is $k\geq1$ units, then we draw two vertical bars $s'_d$ and $s'_b$, such that the end-point of $s'_d$ (resp., $s'_b$) coincides with the rightmost end-point of $s_d$ (resp., leftmost end-point of $s_b$), and such that the other end-point is $k/2$ units above it (resp., below it). In both cases, the two resulting L-shapes see each other through a horizontal visibility segment. See also Fig.~\ref{fi:lvisrep}. Since every vertex is adjacent to at most one crossed edge  we have that the final drawing $\Gamma$ is a L-visibility drawing of $G$. Since $\Gamma'$ contains $O(n)$ segments which have to be transformed to L-shapes, $\Gamma$ is computed in $O(n)$ time. Finally, in order to restore integer coordinates, we scale by a factor 4 the grid of $\Gamma$, which thus takes $O(n^2)$ area.
\end{proof}

Lemmas~\ref{le:st-orientation} and~\ref{le:drawing} imply Theorem~\ref{th:main}. To prove Corollary~\ref{co:main}, consider a visibility drawing $\Gamma$ of an $n$-vertex IC-plane graph $G$. By Theorem~\ref{th:main}, $\Gamma$ can be computed in $O(n)$ time and fits on a grid of $O(n^2)$ size. Let $\ell$ be an L-shape of $\Gamma$. The \emph{representative point} $r$ of $\ell$ is defined as follows. If both the horizontal and the vertical segments of $\ell$ have non-zero length, then $r$ is the point where they touch. Otherwise, $r$ is the midpoint of the segment of $\ell$  having non-zero length. Replace each vertical visibility segment with a polyline as follows. Let $s$ be a visibility segment connecting the two L-shapes $\ell_1$ and $\ell_2$. Let $r_1$ and $r_2$ be the representative points of $\ell_1$ and $\ell_2$, respectively.  Also, let $p_1$ and $p_2$ be the points that $s$ shares with $\ell_1$ and $\ell_2$, respectively. Suppose that $r_1$ is below $r_2$ (and thus $p_1$ is below $p_2$). Replace $s$ with the polyline starting at $r_1$, bending 0.25 grid units above $p_1$, bending again 0.25 units below $p_2$, and ending at $r_2$. With a symmetric operation we can also replace each horizontal visibility segment. Finally, replace each L-shape with its representative point. The resulting drawing is an IC-plane drawing of $G$ where edges are polylines with (at most) two bends that cross at right-angles. Finally, scaling by a factor 4 the grid of the drawing we restore integer coordinates. Fig.~\ref{fi:example-3} shows a RAC drawing computed from the L-visibility drawing in Fig.~\ref{fi:example-2}.

\section{Conclusions and Open Problems}\label{se:conclusions}
We have proved that every IC-plane graph $G$ has a L-visibility drawing which can be computed in linear time. As a corollary, our result implies that $G$ has a RAC drawing in quadratic area and at most two bends per edge which can also be computed in linear time. We conclude the paper with two open problems: $(i)$ Does every $1$-planar graph admit a visibility drawing where the shape associated with each vertex is either an L-shape, a T-shape, or a +-shape? $(ii)$ Does every IC-plane graph admit a RAC drawing with at most one bend per edge in polynomial area?





{\small \bibliography{lshapes}}

\begin{thebibliography}{10}

\bibitem{DBLP:journals/jgaa/ArgyriouBS12}
E.~N. Argyriou, M.~A. Bekos, and A.~Symvonis.
\newblock The straight-line {RAC} drawing problem is {NP-Hard}.
\newblock {\em J. Graph Algorithms Appl.}, 16(2):569--597, 2012.

\bibitem{DBLP:journals/jgaa/Brandenburg14}
F.~J. Brandenburg.
\newblock 1-visibility representations of 1-planar graphs.
\newblock {\em J. Graph Algorithms Appl.}, 18(3):421--438, 2014.

\bibitem{bdek+-rdicg-15}
F.~J. Brandenburg, W.~Didimo, W.~S. Evans, P.~Kindermann, G.~Liotta, and
  F.~Montecchiani.
\newblock Recognizing and drawing {IC}-planar graphs.
\newblock In {\em {GD} 2015}, LNCS. Springer, 2015.
\newblock to appear.

\bibitem{DBLP:journals/jgaa/DeanEGLST07}
A.~M. Dean, W.~S. Evans, E.~Gethner, J.~D. Laison, M.~A. Safari, and W.~T.
  Trotter.
\newblock Bar k-visibility graphs.
\newblock {\em J. Graph Algorithms Appl.}, 11(1):45--59, 2007.

\bibitem{dett-gd-99}
G.~{Di Battista}, P.~Eades, R.~Tamassia, and I.~G. Tollis.
\newblock {\em Graph Drawing}.
\newblock Prentice Hall, 1999.

\bibitem{DBLP:journals/tcs/DidimoEL11}
W.~Didimo, P.~Eades, and G.~Liotta.
\newblock Drawing graphs with right angle crossings.
\newblock {\em Theor. Comput. Sci.}, 412(39):5156--5166, 2011.

\bibitem{dl-cargd-12}
W.~Didimo and G.~Liotta.
\newblock The crossing angle resolution in graph drawing.
\newblock In J.~Pach, editor, {\em Thirty Essays on Geometric Graph Theory}.
  Springer, 2012.

\bibitem{DBLP:journals/jgaa/Evans0LMW14}
W.~S. Evans, M.~Kaufmann, W.~Lenhart, T.~Mchedlidze, and S.~K. Wismath.
\newblock Bar 1-visibility graphs vs. other nearly planar graphs.
\newblock {\em J. Graph Algorithms Appl.}, 18(5):721--739, 2014.

\bibitem{elm-svrp+-15}
W.~S. Evans, G.~Liotta, and F.~Montecchiani.
\newblock Simultaneous visibility representations of plane $st$-graphs using
  {L}-shapes.
\newblock In {\em {WG} 2015}, LNCS. Springer, 2015.
\newblock to appear.

\bibitem{Even1976339}
S.~Even and R.~E. Tarjan.
\newblock Computing an st-numbering.
\newblock {\em Theor. Comput. Sci.}, 2(3):339--344, 1976.

\bibitem{DBLP:conf/apvis/HuangHE08}
W.~Huang, S.-H. Hong, and P.~Eades.
\newblock Effects of crossing angles.
\newblock In {\em PacificVis 2008}, pages 41--46. IEEE, 2008.

\bibitem{ks-cpgic-JGT10}
D.~Kr{\'{a}}l and L.~Stacho.
\newblock Coloring plane graphs with independent crossings.
\newblock {\em J. Graph Theory}, 64(3):184--205, 2010.

\bibitem{DBLP:conf/ictcs/Liotta14}
G.~Liotta.
\newblock Graph drawing beyond planarity: some results and open problems.
\newblock In {\em {ICTCS} '14.}, pages 3--8, 2014.

\bibitem{DBLP:journals/dcg/RosenstiehlT86}
P.~Rosenstiehl and R.~E. Tarjan.
\newblock Rectilinear planar layouts and bipolar orientations of planar graphs.
\newblock {\em Discr. {\&} Comput. Geom.}, 1:343--353, 1986.

\bibitem{TamassiaTollis86}
R.~Tamassia and I.~G. Tollis.
\newblock A unified approach to visibility representations of planar graphs.
\newblock {\em Discr. \& Comput. Geom.}, 1(1):321--341, 1986.

\bibitem{t-rdg-JGT88}
C.~Thomassen.
\newblock Rectilinear drawings of graphs.
\newblock {\em J. Graph Theory}, 12(3):335--341, 1988.

\bibitem{zl-spgic-CEJM13}
X.~Zhang and G.~Liu.
\newblock The structure of plane graphs with independent crossings and its
  applications to coloring problems.
\newblock {\em Central Europ. J. Math.}, 11(2):308--321, 2013.

\end{thebibliography}
\bibliographystyle{abbrv}

\end{document}